\documentclass{article}

\usepackage{amsthm,amsmath}
\usepackage{graphicx}
\usepackage[footnotesize]{caption}

\usepackage[autolanguage]{numprint}
\npfourdigitsep
\newcommand{\num}[1]{\numprint{#1}}

\newtheorem{theorem}{Theorem}
\newtheorem{corollary}{Corollary}
\newtheorem{lemma}{Lemma}

\usepackage{fullpage}

\begin{document}



\title{Estimation of the True Evolutionary Distance\\under the Fragile Breakage Model\thanks{The work is supported by the National Science Foundation under the grant No. IIS-1462107.}}

\author{Nikita Alexeev\qquad and\qquad Max A. Alekseyev}

\date{\small The George Washington University\\Washington, DC, USA}

\maketitle



\begin{abstract} 
\textbf{Background:} The ability to estimate the evolutionary distance between extant genomes plays a crucial role in many phylogenomic studies. Often such estimation is based on the parsimony assumption, implying that the distance between two genomes can be estimated as the \emph{rearrangement distance} equal the minimal number of genome rearrangements required to transform one genome into the other. However, in reality the parsimony assumption may not always hold, emphasizing the need for estimation that does not rely on the rearrangement distance. 
The distance that accounts for the actual (rather than minimal) number of rearrangements between two genomes is often referred to as the \emph{true evolutionary distance}.
 While there exists a method for the true evolutionary distance estimation, it however assumes that genomes can be broken by rearrangements equally likely at any position in the course of evolution. This assumption, known as the \emph{random breakage model}, has recently been refuted in favor of the more rigorous \emph{fragile breakage model} postulating that only certain ``fragile'' genomic regions are prone to rearrangements. 

\textbf{Results:}
We propose a new method for estimating the true evolutionary distance between two genomes under the fragile breakage model. We evaluate the proposed method on simulated genomes, which show its high accuracy. We further apply the proposed method for estimation of evolutionary distances within a set of five yeast genomes and a set of two fish genomes.

\textbf{Conclusions:} 
The true evolutionary distances between the five yeast genomes estimated with the proposed method reveals that some pairs of yeast genomes violate the parsimony assumption. The proposed method further demonstrates that the rearrangement distance between the two fish genomes underestimates their evolutionary distance by about $20\%$. These results demonstrate how drastically the two distances can differ and justify the use of true evolutionary distance in phylogenomic studies.

\end{abstract}




\section*{Background}
Genome rearrangements are evolutionary events that shuffle genomic architectures. Most frequent genome rearrangements are \emph{reversals} (that flip segments of
a chromosome), \emph{translocations} (that exchange segments of two chromosomes), \emph{fusions} (that merge two chromosomes into one), and \emph{fissions} (that split a single chromosome into two). 
These four types of rearrangements can be modeled by Double-Cut-and-Join (DCJ) operations~\cite{yancopoulos2005},
which break a genome at two positions and glue the resulting fragments in a new order.

The ability to estimate the evolutionary distance between extant genomes plays a crucial role in many phylogenomic studies. Often such estimation is based on the parsimony assumption, implying that the distance between two genomes can be estimated as the \emph{rearrangement distance} equal the minimal number of genome rearrangements required to transform one genome into the other. However, in reality the parsimony assumption may not always hold, emphasizing the need for estimation that does not rely on the (minimal) rearrangement distance. The evolutionary distance that accounts for the actual (rather than minimal) number of rearrangements between two genomes is often referred to as the \emph{true evolutionary distance}.

While there exists a method for estimation of the true evolutionary distance~\cite{Lin08}, it however implicitly assumes that genomes can be broken by rearrangements equally likely at any position in the course of evolution. This assumption, known as the \emph{random breakage model} (RBM) of chromosome evolution~\cite{Ohno70,Nadeau84}, has been refuted in favor of the more rigorous \emph{fragile breakage model}
(FBM)~\cite{PevznerTeslerPNAS03} postulating that only certain ``fragile'' genomic regions are prone to rearrangements. The FBM is supported by many recent studies of various genomes (e.g., see references in \cite{Alekseyev10b}). The RBM can be viewed as an extremal case of the FBM, where every genomic region is fragile.

In the current study, we propose a new method for estimating the evolutionary distance between two genomes with high accuracy under the FBM. We assume that the given genomes are represented as sequences of the same blocks (synteny blocks or orthologous genes) and evolved from a common ancestral genome with a number of DCJs. 
Our method estimates the total number of DCJs on the evolutionary path between such genomes. The results of using our method on a simulated dataset show a high level of precision. We also analyze yeast genome data and show that some, but not all pairs of yeast genomes fall under the parsimony assumption.  

The subtle difference between the RBM and FBM from the perspective of true evolutionary distance estimation is the (in)ability to count the number of fragile genomic regions. 
While \emph{breakpoints} (block adjacencies present in one genome and absent in the other) definitely represent fragile regions, 
the shared block adjacencies are treated differently under the two models. 
Namely, under the RBM a shared block adjacency still represents a fragile region, which just happened to remain conserved across the two genomes by chance. Under the FBM, a shared block adjacency may or may not be fragile.

\section*{Methods}\label{sec:backgroud}

\subsection*{Breakpoint Graphs and DCJs}
We start our analysis with circular genomes (i.e., genomes with circular chromosomes) and address linear genomes later.
We represent a genome with $n$ blocks as a \emph{genome graph} composed of $n$ directed \emph{block edges} encoding blocks and their strands, and $n$ undirected \emph{adjacency edges} encoding adjacencies between blocks.

\begin{figure}[!t]
 \centering\includegraphics[width=0.6\textwidth]{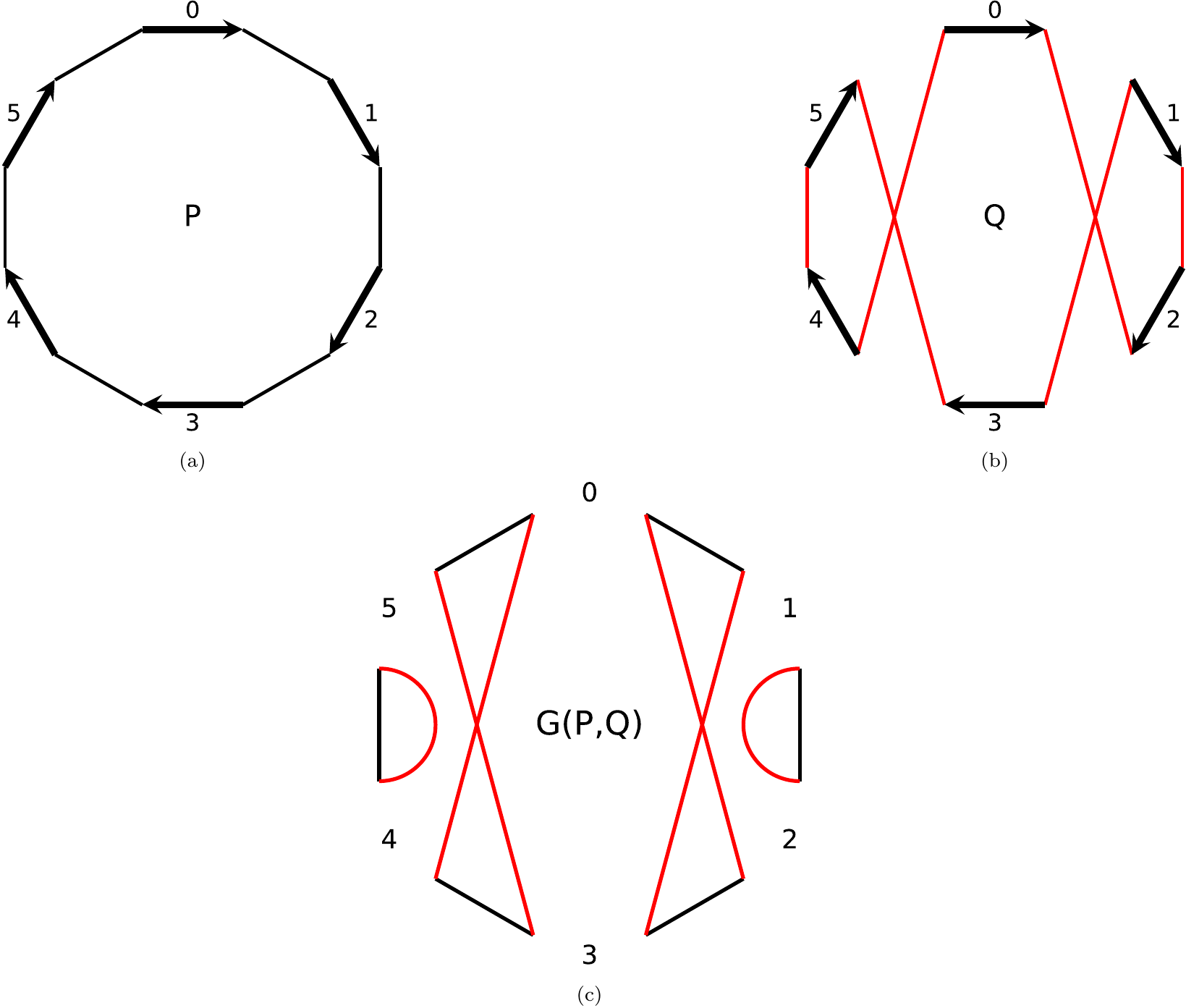}
  \caption{Genome graph of unichromosomal genome $P=(0,1,2,3,4,5)$ with adjacency edges colored black, genome graph of unichromosomal genome $Q = (0,-2,-1,3,-5,-4)$ with adjacency edges colored red,
and their breakpoint graph $G(P,Q)$ of genomes $P$ and $Q$ represents a collection of black-red cycles.}
  \label{fig:genbpg}
\end{figure}

Let $P$ and $Q$ be genomes on the same set of blocks. We assume that in their genome graphs, the adjacency edges of $P$ are colored black (Fig.~\ref{fig:genbpg}a) and the adjacency edges of $Q$ are colored red (Fig.~\ref{fig:genbpg}b).
The \emph{breakpoint graph} $G(P, Q)$ is the superposition of the genome graphs of $P$ and $Q$ with the block edges removed (Fig.~\ref{fig:genbpg}c). 
The black and red adjacency edges in $G(P,Q)$ form a collection of alternating black-red cycles. 

We say that a black-red cycle is an \emph{$\ell$-cycle} if it contains $\ell$ black edges (and $\ell$ red edges) and let $c_{\ell}(P,Q)$ be the number of $\ell$-cycles in $G(P,Q)$. We refer to $1$-cycles as \textit{trivial}\footnote{In the breakpoint graph 
constructed on synteny blocks of two genomes, there are no trivial cycles since no adjacency 
is shared by both genomes. However, the breakpoint graph constructed on orthologous genes or multi-genome synteny blocks may contain trivial cycles.}
and to the other cycles as \emph{non-trivial}. The vertices of non-trivial cycles are called \textit{breakpoints}.

A \emph{DCJ} in genome $Q$ replaces any pair of red adjacency edges $\{x,y\}$, $\{u,v\}$ with either a pair of edges $\{x,u\}$, $\{y,v\}$ or a pair of edges 
$\{u,y\}$, $\{v,x\}$. We say that such a DCJ \emph{operates} on the edges $\{x,y\}$, $\{u,v\}$ and their endpoints $x,y,u,v$.
A DCJ in genome $Q$ transforming it into a genome $Q'$ corresponds to a transformation of the breakpoint graph $G(P,Q)$ into the breakpoint graph $G(P,Q')$ (Fig.~\ref{fig:dcj}).
Each DCJ in the breakpoint graph can merge two black-red cycles into one (if edges $\{x,y\}$, $\{u,v\}$ belong to distinct cycles), split one cycle into two or keep the number of cycles intact (if edges $\{x,y\}$, $\{u,v\}$ belong to the same cycle).
The \emph{DCJ distance} between genomes $P$ and $Q$ is the minimum number of DCJs required to transform $Q$ into $P$. It can be evaluated as $d(P,Q) = b(P,Q) - c(P,Q)$, 
where $b(P,Q)=\sum_{\ell\geq 2} \ell\cdot c_\ell(P,Q)$ is half the number of breakpoints and $c(P,Q)=\sum_{\ell\geq 2} c_\ell(P,Q)$ is the number of non-trivial cycles in the breakpoint graph $G(P,Q)$~\cite{yancopoulos2005}.

\begin{figure}[!t]
 \centering\includegraphics[width=0.8\textwidth]{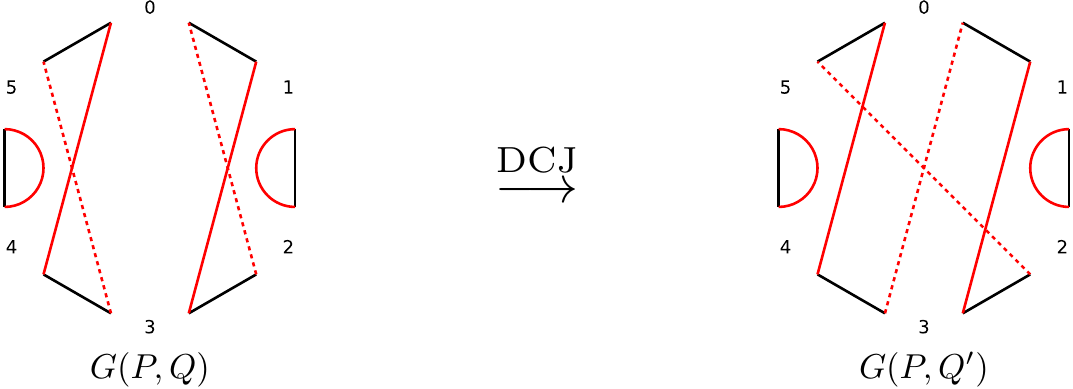}
  \caption{A DCJ in genome $Q$ replaces a pair of red edges in the breakpoint $G(P,Q)$ with another pair of red edges forming matching on the same $4$ vertices.}
  \label{fig:dcj}
\end{figure}

\subsection*{Evolutionary Model}

To estimate the true evolutionary distance between genomes $P$ and $Q$ on the same set of blocks, 
we view the evolution between them as a discrete Markov process that transforms genome $P$ into genome $Q$ with a sequence of DCJs occurring independently. 
The process starts at genome $X=P$ and ends at $X=Q$, and corresponds to a transformation of the breakpoint graphs starting at $G(P,P)$ (formed by a collection of trivial cycles) and ending at $G(P,Q)$. The number of DCJs $k$ in this transformation represents the true evolutionary distance between genomes $P$ and $Q$.

We remark that under the FBM, the number of trivial cycles (if any) in $G(P,Q)$ is an obscure parameter, 
since they may correspond to solid (non-fragile) regions as well as to fragile regions 
that just happen to remain conserved (in both $P$ and $Q$) by chance.\footnote{In contrast, the method of Lin and Moret~\cite{Lin08} considers only the latter option, which corresponds to the RBM.}
Furthermore, there may exist such conserved fragile regions within the blocks forming $P$ and $Q$ and thus such regions cannot be observed in $G(P,Q)$. 
To account for all these possibilities, we 
assume that $P$ and $Q$ are composed of a large unknown number $n$ of solid regions interspersed with the same number of fragile regions, 
some of which remain conserved by chance. In other words, we assume that the given blocks of genomes $P$ and $Q$ are formed by (invisible) solid regions and conserved fragile regions.
Let $P_n$ and $Q_n$ denote these representations of $P$ and $Q$ as sequences of the solid regions. 
We view genome $Q_n$ as obtained from $P_n$ with a sequence of $k$ DCJs each operating on two randomly selected fragile regions (i.e., adjacencies between solid regions).

The crucial observation is that while we do not know the number $n$ of solid regions, 
the breakpoint graphs $G(P,Q)$ and $G(P_n,Q_n)$ have the same cycle structure, except for trivial cycles. 
That is, we have $c_\ell(P_n,Q_n)=c_\ell(P,Q)$ for every $\ell\geq 2$, 
implying, in particular, that $b(P_n,Q_n)=b(P,Q)$, $c(P_n,Q_n)=c(P,Q)$, and $d(P_n,Q_n)=d(P,Q)$.
Indeed, if genomes $P'$ and $Q'$ are obtained from $P$ and $Q$ by replacing a single block $a$ with two consecutive smaller blocks $a_1,a_2$, then $G(P',Q')$ can be obtained from $G(P,Q)$ 
by adding one trivial cycle (corresponding to the shared adjacency $a_1,a_2$). Since the genomes $P_n$ and $Q_n$ can be obtained from $P$ and $Q$ with a number of such operations, the breakpoint graphs $G(P,Q)$ and $G(P_n,Q_n)$ may differ only in the number of trivial cycles.\footnote{In contrast to $c_1(P_n,Q_n)$, the value of $c_1(P,Q)$ is rather arbitrary and thus is ignored in our model.}

In our evolutionary model, the following parameters are observable:
\begin{itemize}
 \item $c_{\ell}=c_{\ell}(P_n,Q_n)=c_{\ell}(P,Q)$ for any $\ell \geq 2$, i.e., the number of $\ell$-cycles in $G(P,Q)$;
 \item $b=b(P_n,Q_n) = b(P,Q) = \sum_{\ell \geq 2} \ell\cdot c_{\ell}$, the number of broken fragile regions between $P$ and $Q$, 
which is also the number of synteny blocks between $P$ and $Q$, or half of the total length of all non-trivial cycles in $G(P,Q)$;
 \item $d = d(P_n,Q_n) = d(P,Q) =  b - \sum_{\ell\geq 2} c_\ell$, the DCJ distance between $P$ and $Q$;
\end{itemize}
while the following parameters are hidden:
\begin{itemize}
\item $c_1 = c_1(P_n,Q_n)$, the number of trivial cycles in $G(P_n,Q_n)$; 
 \item $n=n(P)=n(Q)$, the number of fragile regions in each of genomes $P$ and $Q$, half the total length of all cycles in $G(P_n,Q_n)$;
 \item $k = k(P,Q)$, the number of DCJs in the Markov process, the true evolutionary distance between $P$ and $Q$.
\end{itemize}

\subsection*{Extension to linear genomes}
To analyze linear genomes, we add one artificial adjacency between telomeres for each chromosome and consider the resulting circular genomes.
Since the number of chromosomes is often negligible as compared to the number of fragile regions (for example, in the yeast genomes that we analyze in the \emph{Discussion} section, the number of chromosomes ranges between $6$ and $8$, while the number of fragile regions is at least $710$), these artificial edges do not significantly affect the estimation. We refer to \cite{alekseyev2008multi} for a discussion of subtle differences in the analysis of circular and linear genomes.

\section*{Results}
We propose a new method for estimating the evolutionary distance between two genomes with high accuracy under the FBM.
A key component of our method is the analytical estimation of $c_\ell$ for various $\ell\geq 2$ in terms of $k$ and $n$, which we describe in the \emph{Theoretical Analysis} subsection below.
From this estimation, in the \emph{Estimation Algorithm} subsection we solve the inverse problem to find the true evolutionary distance $k$ (which is our goal) 
and the number of fragile regions $n$ (as a by-product). 
In our analysis, we consider only relatively small $\ell$ and assume that $n$ and $k$ are sufficiently large (see Theorem \ref{thm:main}).

\subsection*{Theoretical Analysis}
\label{sec:thana}

\begin{theorem}
Let genome $P_n$ be a genome with $n$ fragile regions and genome $Q_n$ be obtained from $P_n$ with $k = \lfloor{\gamma n /2}\rfloor$ random DCJs for some $\gamma>0$. 
Then, for any fixed $\ell$, the proportion of edges that belong to $\ell$-cycles in $G(P_n,Q_n)$ is
 $$\frac{\ell c_\ell}{n} = e^{-\gamma\ell}\frac{(\gamma\ell)^{\ell-1}}{\ell!} +O_p\left(\frac{1}{\sqrt{n}}\right)\textrm{ as } n \to \infty \, ,$$
 where $O_p\left(\frac{1}{\sqrt{n}}\right)$ denotes a term stochastically bounded\footnote{We remind that a sequence of random variables 
 $\{X_n\}$ is \emph{stochastically bounded} by a deterministic sequence $\{a_n\}$, denoted $X_n = O_p(a_n)$, if for all $\varepsilon>0$, there exists $C$ such that for all $n$, $\mathrm{Pr}\{\left|X_n/a_n\right|>C\}<\varepsilon$. Similarly, $X_n = o_p(a_n)$ if for all $C>0$, $\lim_{n\to\infty} \mathrm{Pr}\{\left|X_n/a_n\right| > C\} =0$.}.
 \label{thm:main}
\end{theorem}

To prove Theorem~\ref{thm:main}, we will need some lemmas.

Let us consider a sequence $\mathcal{D}$ of $k$ random DCJs transforming the graph $G(P_n,P_n)$ into the graph $G(P_n,Q_n)$.
Each DCJ in $\mathcal{D}$ either merges two cycles (\emph{merging DCJ}), splits one cycle into two  (\emph{splitting DCJ}), or preserves the cycle structure (\emph{preserving DCJ}).
We call set of $\ell$ black edges \emph{proper} for the transformation $\mathcal{D}$ if in the breakpoint graph $G(P_n,Q_n)$ these edges form an $\ell$-cycle resulted entirely from $\ell-1$ merging DCJs. 
We remark that each $\ell$-cycle in $G(P_n,Q_n)$ is a result of $\ell-1$ merging DCJs (we denote the number of such cycles by $\tilde{c}_\ell$) 
or is formed with at least one splitting or preserving DCJ. 
The following lemmas provide an asymptotic for $\tilde{c}_\ell$ and an upper bound for $c_{\ell} - \tilde{c}_\ell$.

\begin{lemma}\label{lem:mean}
The mean value of $\tilde{c}_\ell$ has the following asymptotic:
$$\mathrm{E}(\tilde{c}_\ell) = \frac{n}{\ell} e^{-\gamma\ell}\frac{(\gamma\ell)^{\ell-1}}{\ell!} + O(1) \, .$$
\end{lemma}

\begin{proof}
Let us fix a set $A$ of $\ell$ black edges and find the probability that $A$ is proper for $\mathcal{D}$. 
We remark that this probability does not depend on the content of $A$ but only on its size $|A|=\ell$, and denote it by $\tilde{p}_{n,k,\ell}$.
There are $2^{\ell-1} (\ell-1)!$ ways to arrange $\ell$ black edges into an $\ell$-cycle, since there exist $(\ell-1)!$ circular permutations of length $\ell$ and $2^{\ell-1}$ ways to assign directions to their elements. For any $\ell$-cycle, there are $\ell^{\ell-2}$ ways to obtain it as a result of $\ell-1$ DCJs~\cite{ouangraoua2010}.

We represent $\mathcal{D}$  as the union $\mathcal{S} \cup \mathcal{\bar{S}}$, where the subsequence $\mathcal{S}$ contains $\ell-1$ DCJs connecting edges from $A$ into an $\ell$-cycle, and the subsequence $\mathcal{\bar{S}}$ contains the rest of the DCJs. Since $|\mathcal{D}|=k$ and $|\mathcal{S}|=\ell-1$, there are $\binom{k}{\ell-1}$ ways to choose positions for elements of $\mathcal{S}$ in $\mathcal{D}$. 
The $k-\ell+1$ DCJs from $\mathcal{\bar{S}}$ operate on the $n-\ell$ red edges that are not incident to any black edge from $A$, 
and for each pair of red edges there are two possible ways to recombine them with a DCJ. This gives
$$2^{k-\ell+1}\binom{n-\ell}{2}^{k-\ell+1} $$
possible subsequences $\mathcal{\bar{S}}$.
So, there are 
$$2^{\ell-1}(\ell-1)! \ell^{\ell-2} \binom{k}{\ell-1} 2^{k-\ell+1} \binom{n-\ell}{2}^{k-\ell+1}$$
transformations $\mathcal{D}$ such that $A$ is proper for $\mathcal{D}$.
The total number of $k$-step transformations is equal to
  $2^k \binom{n}{2}^k \, .$
Then the probability $\tilde{p}_{n,k,\ell}$ can be found as follows
 \begin{equation}\label{eq:p}
 \tilde{p}_{n,k,\ell} = \frac{\ell^{\ell-2}(\ell-1)!\binom{k}{\ell-1}\binom{n-\ell}{2}^{k-\ell+1}}{\binom{n}{2}^k} \,.
 \end{equation}

Since there are $\binom{n}{\ell}$ ways to choose the set $A$ of black edges, the average number $\mathrm{E}(\tilde{c}_\ell)$ of $\ell$-cycles obtained by only merging DCJs equals 
\begin{equation}\label{eq:binp}
  \begin{split}
  \binom{n}{\ell} \tilde{p}_{n,k,\ell} &= \ell^{\ell-2}(\ell-1)!\binom{n}{\ell}\binom{k}{\ell-1} \cdot\frac{\binom{n-\ell}{2}^{k-\ell+1}}{\binom{n}{2}^k} \\
  &= \ell^{\ell-2}(\ell-1)! \, \frac{n^\ell}{\ell!}\cdot\frac{k^{\ell-1}}{(\ell-1)!}\frac{(n-\ell)^{2(k-\ell+1)}}{2^{k-\ell+1}}\cdot \frac{2^k}{n^{2k}}+O(1)  \\
  &= \ell^{\ell-2} k^{\ell-1}\,\cdot \frac{n^\ell}{\ell!}\cdot \frac{2^{\ell-1}}{n^{2\ell - 2}} \left(1-\frac{\ell}{n}\right)^{2(k-\ell+1)}+O(1)  \\
  &= \left(\frac{2k\ell}{n}\right)^{\ell-1} \frac{n}{\ell}\cdot \frac{1}{\ell!}e^{-\gamma \ell}+O(1) \\
  &= \frac{n}{\ell} e^{-\gamma\ell}\frac{(\gamma\ell)^{\ell-1}}{\ell!} + O(1) \, .
  \end{split}
\end{equation}

\end{proof}

\begin{lemma}
  The variance of $\tilde{c}_\ell$ is bounded: $\mathrm{Var} \left(\tilde{c}_\ell\right) = O(n)$. 
  \label{lem:var}
\end{lemma}

\begin{proof}
For a transformation $\mathcal{D}$ and a set $A$ of black edges, we define a random variable $X_A$: 
$$
  X_{A} = \begin{cases}
                   1, & \textrm{if }A\textrm{ is proper for } \mathcal{D},\\
                   0, & \textrm{otherwise.}
                  \end{cases}
$$
By definition, $\tilde{c}_\ell = \sum_A X_A,$ where the sum is taken over all possible sets of $\ell$ black edges. The variance of $\tilde{c}_\ell$ is equal to
$$\mathrm{E} \left(\tilde{c}_\ell^2\right) - \mathrm{E} \left(\tilde{c}_\ell\right)^2 \, .$$
Since
$$\mathrm{E} \left(\tilde{c}_\ell^2\right) = \sum_A \sum_B \mathrm{E} \left(X_A X_B\right) \, ,$$
we will find $\mathrm{E} \left(X_A X_B\right)$ for each pair $A,B$ of sets of $\ell$ black edges:
$$
  \mathrm{E} \left(X_A X_B\right) = \begin{cases}
                   \tilde{p}_{n,k,\ell}, & \textrm{if }A=B \, ,\\
                   \tilde{p}_{n,k,\ell}\cdot\tilde{p}_{n-\ell,k-\ell+1,\ell}, & \textrm{if }A\cap B = \emptyset \, ,\\
                   0, & \textrm{otherwise,}
                  \end{cases}
$$
where $\tilde{p}_{n,k,\ell}$ is defined by \eqref{eq:p}.
This implies 
$$\mathrm{E} \left(\tilde{c}_\ell^2\right) = \binom{n}{\ell} \tilde{p}_{n,k,\ell} +\binom{n}{\ell} \binom{n-\ell}{\ell} \tilde{p}_{n,k,\ell}\cdot \tilde{p}_{n-\ell,k-\ell+1,\ell} \, .$$
Finally, using \eqref{eq:binp}, we obtain
\begin{small}
\[
\begin{split}
\mathrm{Var}\left(\tilde{c}_\ell\right) &= \binom{n}{\ell} \tilde{p}_{n,k,\ell} +\binom{n}{\ell} \binom{n-\ell}{\ell} \tilde{p}_{n,k,\ell}\cdot \tilde{p}_{n-\ell,k-\ell+1,\ell} - \mathrm{E}  \left(\tilde{c}_\ell\right)^2\\
&=\binom{n}{\ell} \tilde{p}_{n,k,\ell}\left(1+\binom{n-\ell}{\ell}\tilde{p}_{n-\ell,k-\ell+1,\ell}-\binom{n}{\ell} \tilde{p}_{n,k,\ell}\right)\\
&=\left(\frac{n}{\ell} e^{-\gamma\ell}\frac{(\gamma\ell)^{\ell-1}}{\ell!} + O(1)\right)\left(1+\frac{n-\ell}{\ell} e^{-\gamma\ell}\frac{(\gamma\ell)^{\ell-1}}{\ell!} -\frac{n}{\ell} e^{-\gamma\ell}\frac{(\gamma\ell)^{\ell-1}}{\ell!} + O(1)\right) \\
  &= \left(\frac{n}{\ell} e^{-\gamma\ell}\frac{(\gamma\ell)^{\ell-1}}{\ell!} + O(1)\right) \left(1- e^{-\gamma\ell}\frac{(\gamma\ell)^{\ell-1}}{\ell!} + O(1)\right) = O(n).
\end{split}
\]
\end{small}
\end{proof}

For a positive integer $M$, we call a splitting DCJ \emph{$M$-splitting}, if it splits some $m$-cycle into an $i$-cycle and an $(m-i)$-cycle for some $i \leq M$.
Similarly, we call a preserving DCJ \emph{$M$-preserving} if it operates on an $m$-cycle for some $m \leq M$.

\begin{lemma} Let $M$ be any positive integer. Then (i) the number of $M$-splitting DCJs in $\mathcal{D}$ is stochastically dominated by a Poisson random variable with parameter $\gamma M/2$;
(ii) the number of $M$-preserving DCJs in $\mathcal{D}$ is stochastically dominated by a Poisson random variable with parameter $\gamma M^2/4$.
 \label{lem:pois}
\end{lemma}
\begin{proof}
To prove (i), we notice that the probability that a DCJ from $\mathcal{D}$ splits a fixed $m$-cycle into an $i$-cycle and an $(m-i)$-cycle is 
$\frac{m}{n(n-1)}$ (if $i\ne m/2$) or $\frac{m}{2n(n-1)}$ (if $i=m/2$). The probability that a DCJ splits a cycle into an $i$-cycle with $i \leq M$ and another cycle can be bounded by
 $$\sum_{i=1}^M \sum_{m>i} \frac{m}{n(n-1)}c_m \leq \frac{M}{n-1} \, .$$
 This bound implies that a number of $M$-splitting DCJs in $\mathcal{D}$ is stochastically dominated by the random variable $Y$ that equals $j$ with the probability:
 $$
 \binom{k}{j}\left(\frac{M}{n-1}\right)^j \left(1-\frac{M}{n-1}\right)^{k-j}= e^{-\gamma M/2} \frac{(\gamma M/2)^j}{j!} + o(1)\, .
 $$
 Since $e^{-\gamma M/2} \frac{(\gamma M/2)^j}{j!}$ represents the probability that a Poisson random variable with parameter $\gamma M/2$
 is equal to $j$, the proof of (i) is completed.

To prove (ii), we notice that the probability that a DCJ from $\mathcal{D}$ operates on two red edges from the fixed $m$-cycle and does not split this cycle equals $\frac{m(m-1)}{2n(n-1)}$. 
Summing over all $m\leq M$, we bound the probability that a fixed DCJ is $M$-preserving as
 $$\sum_{m \leq M} \frac{m^2}{2n(n-1)}c_m \leq \frac{M^2}{2(n-1)}$$
and the proof is completed similarly to the proof of (i) above.
\end{proof}

Now we can prove Theorem~\ref{thm:main}.
\begin{proof}[Proof of Theorem~\ref{thm:main}.]
 Lemma~\ref{lem:pois} 
implies that $c_\ell - \tilde{c}_\ell$, the number of $\ell$-cycles obtained with at least one splitting or preserving DCJ, is of order $O(1)$. 
Indeed, each such $\ell$-cycle uniquely determines the last splitting or preserving DCJ that participates in the formation of this cycle 
(i.e., operates on some of the cycle vertices). 
Clearly, this last splitting (resp. preserving) DCJ is $\ell$-splitting (resp. $\ell$-preserving). 
Since such last DCJ correspond to at most two cycles, it follows that the number $c_\ell - \tilde{c}_\ell$ does not exceed twice the number of $\ell$-splitting and $\ell$-preserving DCJs.
By Lemma~\ref{lem:pois}, 
the number of $\ell$-splitting and $\ell$-preserving DCJs is bounded (stochastically dominated by a Poisson random variable), 
we have $c_\ell - \tilde{c}_\ell = O(1)$.
 
 Since by Lemmas~\ref{lem:mean} and \ref{lem:var},  the number $\tilde{c}_\ell$ has the mean value $\frac{n}{\ell} e^{-\gamma\ell}\frac{(\gamma\ell)^{\ell-1}}{\ell!} + O(1)$ and the standard deviation of order $O(\sqrt{n})$.
 The fraction of edges in $\ell$-cycles is $\frac{\ell c_\ell}{n}$. Applying Chebyshev's inequality we obtain:
 $$
 \frac{\ell c_\ell}{n} = e^{-\gamma\ell}\frac{(\gamma\ell)^{\ell-1}}{\ell!} +O_p\left(\frac{1}{\sqrt{n}}\right)\, ,
 $$
 which completes the proof.
 \end{proof}

We remark that for $\gamma < 1$, the sequence 
$$p_\ell = e^{-\gamma\ell}\frac{(\gamma\ell)^{\ell-1}}{\ell!},\qquad \ell=1,2,\dots$$
defines a probability mass function, which characterizes a Borel distribution \cite{tanner1961}. So, if $\gamma<1$, 
for a randomly chosen edge, the length of the cycle that contains it follows a Borel distribution with parameter $\gamma$. 
Moreover, the results of Erd{\"o}s and R{\'e}nyi~\cite{erdos1960} imply that
$$1-\sum_{\ell=1}^{\infty}{e^{-\gamma\ell}\frac{(\gamma\ell)^{\ell-1}}{\ell\cdot\ell!}} = \frac{\gamma}{2} \, ,$$
which can also be seen empirically in Fig.~\ref{fig:d2br}a.
It follows that for $\gamma<1$, $d = k(1+o(1))$ and the DCJ distance closely approximates the true distance. 
While $d=k$ corresponds to the parsimony assumption, we say that the process is in the \emph{parsimony phase} as soon as $\gamma<1$.

From Theorem~\ref{thm:main}, we have
\begin{corollary}\label{cor:br}
 $$\frac{b}{n} = 1 - e^{-\gamma} +o_p(1) \, .$$
\end{corollary}
\begin{proof}
 Indeed, $b = n - c_1$ and $\frac{c_1}{n} = e^{-\gamma}+o_p(1)\,$.
\end{proof}

\begin{corollary}\label{cor:d2}
$$\frac{d}{n} = 1-\sum_{\ell=1}^{\infty} \frac{p_{\ell}}{\ell} +o_p(1) \, .$$
\end{corollary}
\begin{proof} 
By definition, 
$$d = b - \sum_{\ell=2}^\infty c_\ell = n - \sum_{\ell=1}^\infty c_\ell \, .$$
For any fixed $M$, the number of $m$-cycles, where $m > M$, is bounded by $\frac{n}{M}$. By Theorem \ref{thm:main}, for any fixed $\ell$, we have $\frac{c_\ell}{n} = \frac{p_\ell}{\ell}+\frac{\xi_\ell}{\sqrt{n}}$, where $\xi_\ell$ is some random variable with $\mathrm{E}(\xi_\ell) =0$ and $\mathrm{Var}(\xi_\ell) < 1$. Hence, for any fixed $M$ we have:
$$\frac{1}{n}\sum_{\ell=1}^\infty c_\ell = \frac{1}{n}\sum_{\ell=1}^M c_\ell + \frac{1}{n}\sum_{\ell=M+1}^\infty c_\ell \leq \sum_{\ell=1}^{M}\left(\frac{p_\ell}{\ell}+\frac{\xi_\ell}{\sqrt{n}}\right)+\frac{1}{M} \, .$$
Let $\eta = \frac{1}{M}\sum_{\ell=1}^M \xi_\ell$,  then
$$\frac{1}{n}\sum_{\ell=1}^\infty c_\ell \leq \sum_{\ell=1}^{M}\frac{p_\ell}{\ell}+\frac{M\eta}{\sqrt{n}}+\frac{1}{M} \, .$$
From the definition of $\eta$ it follows that the random  variable $\eta$ has the mean value $\mathrm{E}(\eta) =0$ and the variance $\mathrm{Var}(\eta) < 1$. 
Let
$$r(m) = \sum_{\ell=m+1}^{\infty}\frac{p_\ell}{\ell} \, ,$$
so that $r(0) = \sum_{\ell=1}^{\infty} \frac{p_{\ell}}{\ell}$.
Then
$$\frac{1}{n}\sum_{\ell=1}^\infty c_\ell - r(0) \leq \frac{1}{M} - r(M) +\frac{M\eta}{\sqrt{n}} \, .$$
The lower bound is obvious
$$\frac{1}{n}\sum_{\ell=1}^\infty c_\ell - r(0) \geq - r(M)+\frac{M\eta}{\sqrt{n}} \, .$$

Since the series $\sum_{\ell=1}^{\infty}\frac{p_\ell}{\ell}$ converges, for each $\varepsilon > 0$ one can choose $M$ and then choose $n$ in such a way that
$$-\varepsilon \leq \frac{1}{n}\sum_{\ell=1}^\infty c_\ell - r(0) \leq \varepsilon(1+\eta) \, .$$
\end{proof}

The estimations given in Corollaries~\ref{cor:br} and~\ref{cor:d2} are very precise, as shown in Fig.~\ref{fig:d2br}. All the simulations are performed for $n = \num{1000}$ (see \emph{Discussion}).

\begin{figure}[!t]
\begin{center}
\begin{tabular}{lll}
\textbf{(a)} & ~~~ & \textbf{(b)} \\
\\
\begin{tabular}{c}
\includegraphics[width=0.45\textwidth]{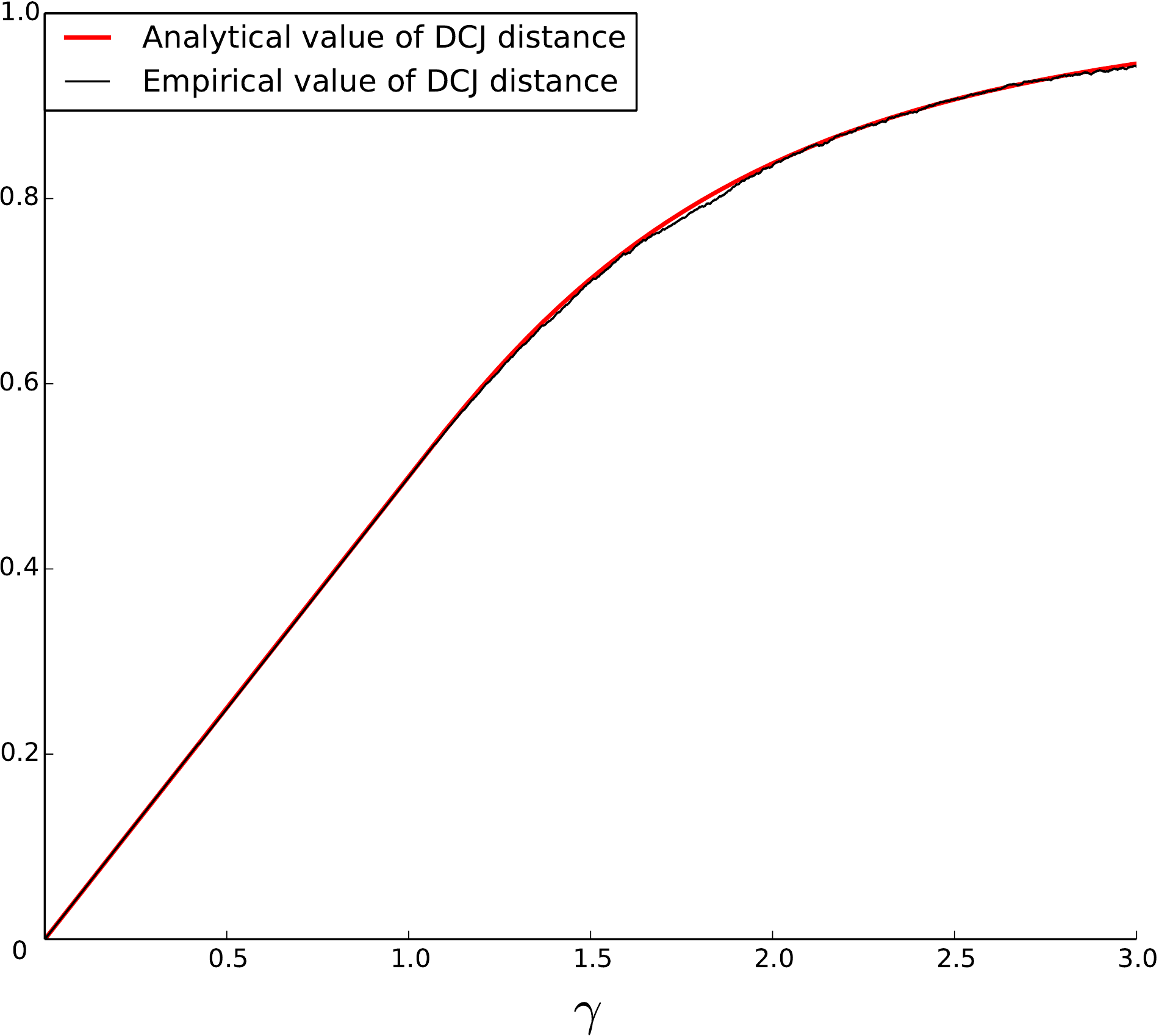}
\end{tabular}
&&
\begin{tabular}{c}
\includegraphics[width=0.45\textwidth]{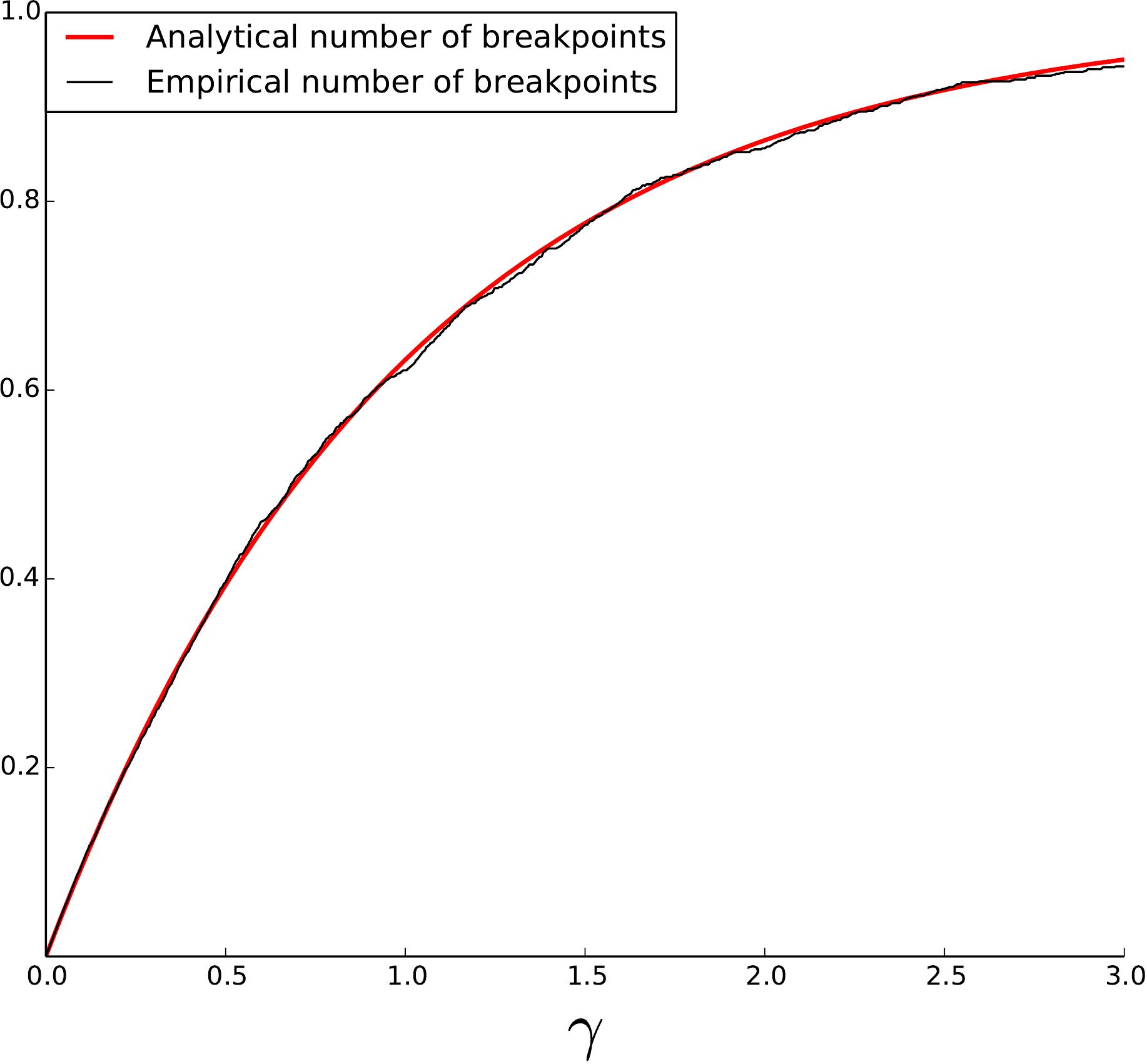}
\end{tabular}
\end{tabular}
\caption{\textbf{(a)} Empirical and analytical curves for the fraction $\frac{d}{n}$ as a function of $\gamma$. 
\textbf{(b)} Empirical and analytical curves for the fraction $\frac{b}{n}$ as a function of $\gamma$.}
\label{fig:d2br}
\end{center}
\end{figure} 

\subsection*{Estimation Algorithm}
\label{sec:estalg}

From Corollaries~\ref{cor:br} and \ref{cor:d2}, we obtain the following approximation for the ratio $\frac{d}{b}$:
\begin{equation}
\label{doverb}
\frac{d}{b} \approx \frac{1-\sum_{\ell=1}^{\infty}{e^{-\gamma\ell}\frac{(\gamma\ell)^{\ell-1}}{\ell\cdot\ell!}}}{1 - e^{-\gamma}} \, ,
\end{equation}
and then estimate $\gamma$ with the bisection method. 
The plot of $\frac{d}{b}$ as a function of $\gamma$ is shown in Fig.~\ref{fig:dtobcyc6}a. As one can see, this function is increasing, and each value of  $\frac{d}{b}$ in the interval $[0.5,1]$ uniquely determines the value of $\gamma$. In particular, $\gamma =1$ corresponds to $\frac{d}{b} \approx 0.79$. That is, if $\frac{d}{b}<0.79$, then the process is in the parsimony phase and the true distance $k$ is accurately approximated by the DCJ distance $d$. 

Our simulations demonstrate that this estimation method is robust for $\gamma < 2$. For larger values of $\gamma$, the estimator is too sensitive to small random deviations.

\begin{figure}[!t]
\begin{center}
\begin{tabular}{lll}
\textbf{(a)} & ~~~ & \textbf{(b)} \\
\\
\begin{tabular}{c}
\includegraphics[width=0.45\textwidth]{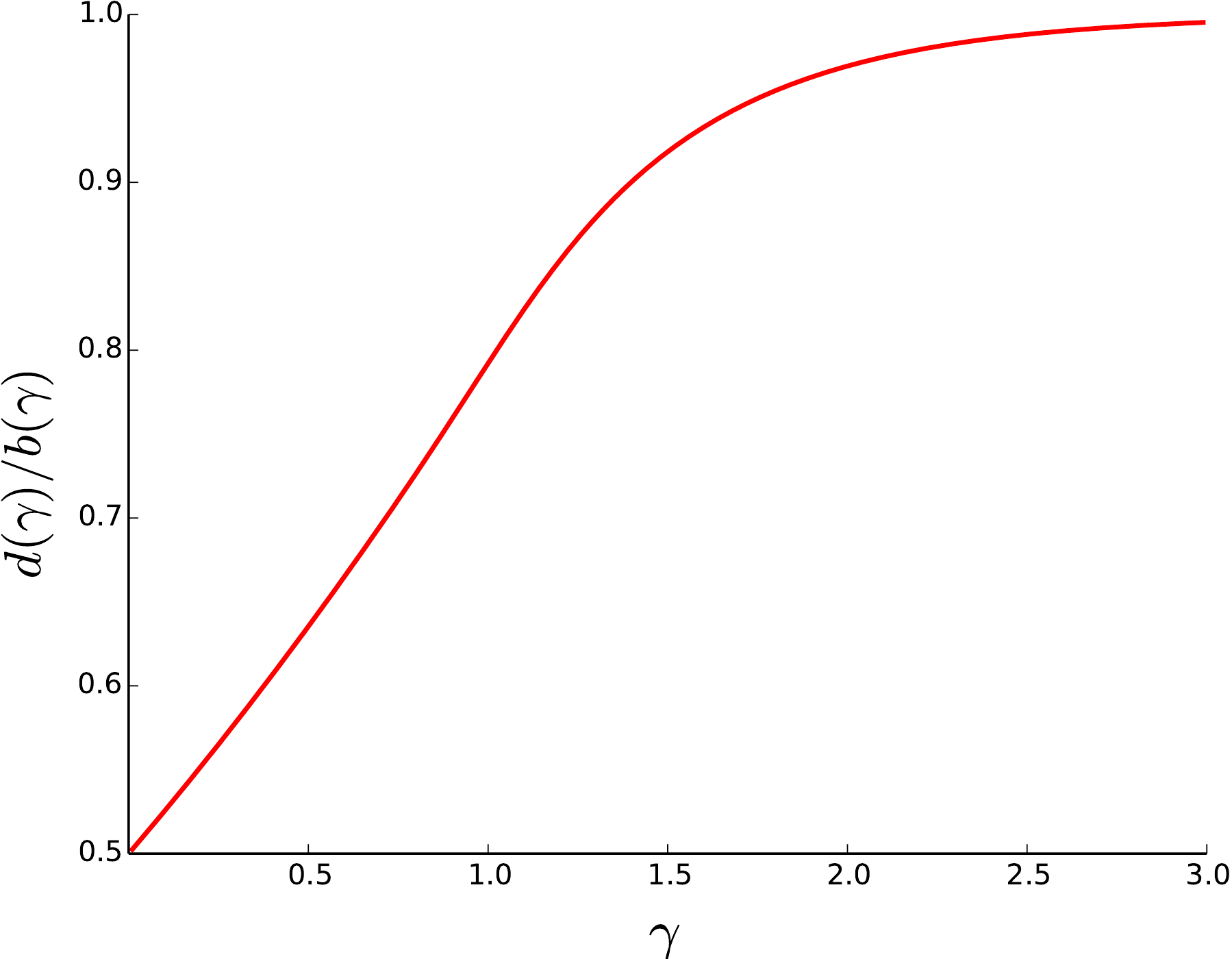}
\end{tabular}
&&
\begin{tabular}{c}
\includegraphics[width=0.45\textwidth]{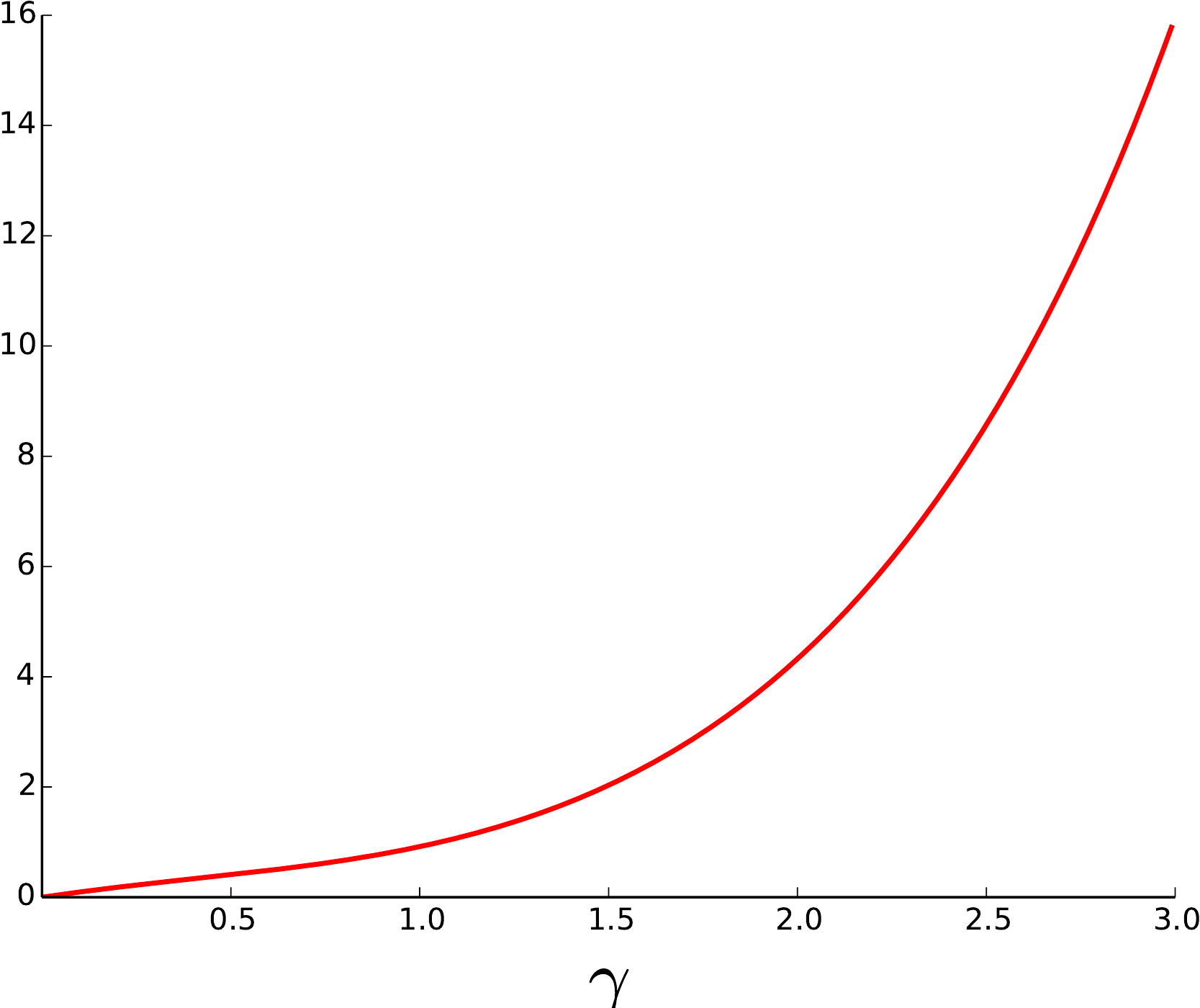}
\end{tabular}
\end{tabular}
\caption{\textbf{(a)} The ratio $\frac{d}{b}$ as a function of $\gamma$. \textbf{(b)} The plot of function $f_6(\gamma)$ defined in \eqref{eq:cyc6}.}
    \label{fig:dtobcyc6}
\end{center}
\end{figure}

Alternatively, the value of $\gamma$ can be approximated by $\frac{b}{\sum_{\ell =2}^{m} \ell c_\ell}$ for some $m$. In our simulations, the best results were observed for $m =6$. Again, the function 

\begin{align}
f_m(\gamma) = \frac{1 - e^{-\gamma}}{\sum_{\ell=2}^{m}e^{-\gamma\ell}\frac{(\gamma\ell)^{\ell-1}}{\ell!}}
\label{eq:cyc6}
\end{align}
is increasing (Fig.~\ref{fig:dtobcyc6}b). The applicability limits for this estimator depend on the value of $n$. 
Namely, if $\sum_{\ell =2}^{m} \ell c_\ell$ is close to zero, then the estimator becomes sensitive to small deviations. 
In our simulations with $n = \num{1000}$ (see \emph{Discussion}), the estimator is very precise for values $\gamma < 2.5$ with the relative error below 10\% in 95\% of observations.

Once we obtain an estimated value $\gamma_e$ for $\gamma$, it is easy to estimate the values of $n$ and $k$ as follows:
\begin{equation}
\label{eq:nkest}
n_{e} = \frac{b}{1-e^{-\gamma_{e}}}\quad\text{and}\quad k_{e} = \frac{\gamma_{e} \cdot n_{e}}2.
\end{equation} 
We report a value of $k_e$ as our estimation for the true evolutionary distance.

\section*{Discussion}
We evaluated the proposed method on simulated genomes, and further applied it for estimation of the evolutionary distances within a set of five yeast genomes and a set of two fish genomes.

\subsection*{Simulated Genomes}

\begin{figure}[!t]
\centering
\centering\includegraphics[width=0.6\textwidth]{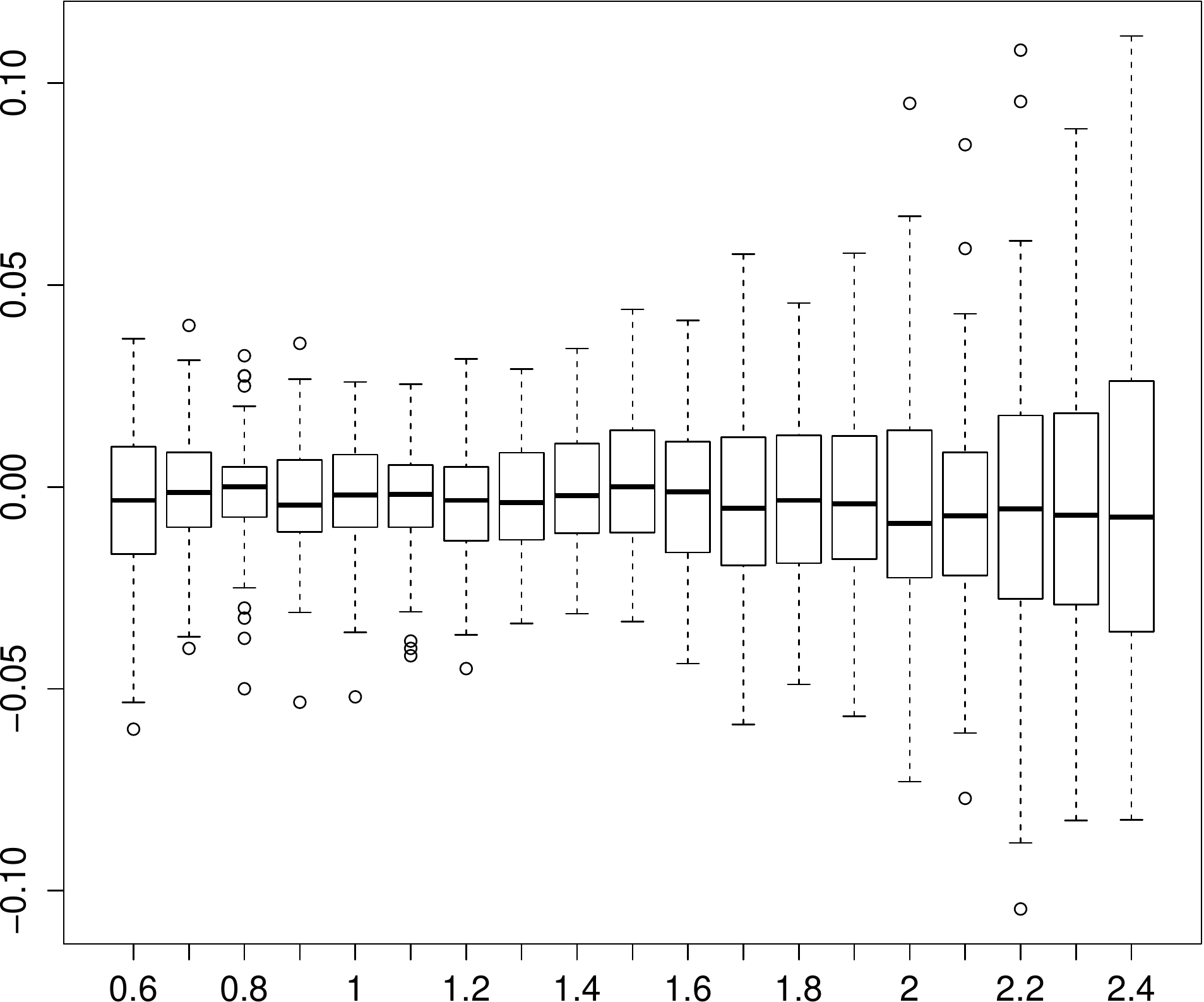}
 \caption{The dependency of distribution of the relative error $\frac{k_{e}-k}{k}$ on $\gamma$.}
 \label{fig:est6}
\end{figure}

We performed a simulation with a fixed number of blocks $n = \num{1000}$ 
and various values of $\gamma$.
In each simulation, we started with a genome $P$ on $n$ blocks and applied a number of DCJs, until we reached the upper value of $\gamma$ (in our case this upper value is 2.5). We denote the resulting genome by $Q$ and estimate $\gamma$, $n$, and $k$ from the genomes $P$ and $Q$. 
First we estimate $\gamma$ by solving the approximate equation
\begin{align}
&\frac{b}{2c_2 +3c_3 +4c_4 +5c_5 +6c_6} \approx f_6(\gamma) \, .
 \label{eq:est6}
\end{align}
Since the function in the r.h.s. is continuous and increasing, a unique solution exists for any value of the l.h.s.
From the estimated value $\gamma_{e}$, we compute $n_{e}$ and $k_{e}$ as in \eqref{eq:nkest}. 

The result of this procedure is shown in Fig.~\ref{fig:est6} with a box-plot diagram of relative error $\frac{k_{e}-k}{k}$ of our estimation for each $\gamma \in [0.6, 2.5]$ with a step $0.1$. As one can see, the relative error increases with the increase of $\gamma$, but even for large values of $\gamma$, the median of the relative error is small (e.g., for $\gamma = 2.5$ the median of the relative error is only $0.0075$); and for all $\gamma$, the interquartile range is less than $0.1$. 

Instead of having $6$ terms in the denominator of \eqref{eq:est6}, one can take a smaller number $m \geq 2$. While the behavior of the sum with a larger $m$ is more stable generally, our simulations showed very close results for all $m$ between $2$ and $10$.

Instead of \eqref{eq:est6}, we also used the approximate equation \eqref{doverb}, which involves only two observable parameters: the DCJ distance $d$ and half the number of breakpoints $b$. While the resulting estimation is quite similar to the previous one, it is more accurate when $\gamma<1.6$ and less accurate $\gamma>2$.

\subsection*{Yeast genomes}\label{sect:yeast}
We analyzed a set of five yeast genomes: \textit{A. gossypii}, \textit{K. lactis}, \textit{K. thermotolerans}, \textit{S. kluyveri}, and \textit{Z. rouxii}, represented 
as sequences of the same $710$ synteny blocks~\cite{Tannier2009}.
For each pair of genomes, we circularized their chromosomes, constructed the breakpoint graph, and independently estimated the true evolutionary distance between them. The results in Table~\ref{tab:yeast} demonstrate that some but not all pairs of yeast genomes fall under the parsimony phase.

\begin{table}[!t]
\caption{For each pair of yeast genomes, $x\,:\,y$ gives the estimated true evolutionary distance $x$ and the rearrangement distance $y$.}
\label{tab:yeast} 
\centering
\begin{tabular}{| c | c | c | c | c | }
\hline
 &   \emph{K. lactis} &    \emph{K. thermotolerans} & \emph{S. kluyveri} & \emph{Z. rouxii}\\
  \hline
  \emph{A. gossypii}  &  375 : 359 & 260 : 249 & 228 : 217 &  336 : 319\\
  \hline
 \emph{K. lactis}  &  & 280 : 272 & 253 : 240 & 352 : 344 \\
  \hline
   \emph{K. thermotolerans}   &   &  & 75 : 75 & 198 : 198  \\
  \hline
  \emph{S. kluyveri}   &  &   &  & 162 : 162 \\
  \hline
\end{tabular}
\end{table}

\subsection*{Fish genomes}\label{sect:fish}
We also analyzed two fish genomes: \textit{Tetraodon nigroviridis}~\cite{jaillon2004} and \textit{Gasterosteus aculeatus}~\cite{jones2012} represented as sequences of the same $\num{6132}$ genes.
Our estimation for the true evolutionary distance between these genomes shows that they do not fall under the parsimony phase:
their rearrangement distance equals $\num{3705}$, while the true evolutionary distance is about $\num{4500}$.

\section*{Conclusions}
In the current study, we address the problem of estimating the true evolutionary distance between two genomes under the fragile breakage model. 
Similarly to our previous study on estimation of the proportion of evolutionary transpositions \cite{alexeev2015transposition}, 
we model evolution as a sequence of random DCJs and track how these random DCJs change the cycle structure of the breakpoint graph. 
We show that, while the number of DCJs is less than half the number of fragile regions in the genome, the parsimony assumption holds, 
and in this case we prove that lengths of alternating cycles are distributed according to a Borel distribution. 
In this sense our process in the parsimony phase is closely related to the evolution of a random Erd{\"o}s--R{\'e}nyi graph~\cite{erdos1960} in  the subcritical regime.
A similar process was also analyzed by Berestycki and Durrett~\cite{berestycki2006}. 
They studied the cycle structure of a permutation obtained from the identity permutation with a sequence of random algebraic transpositions. Our results are consistent with their findings.

We provide estimators for the true evolutionary distance, which show high accuracy on simulated genomes. Our analysis of five yeast genomes shows that some but not all genome pairs fall under the parsimony phase. We also analyzed two fish genomes and revealed that the rearrangement distance between them underestimate the true evolutionary distance by about $20\%$. This data show how drastically the two distances can differ and emphasize the importance of using the evolutionary (rather than rearrangement) distance in comparative genomics.
In contrast to the method of Lin and Moret~\cite{Lin08}, our method does not rely on the number of common gene adjacencies across two given genomes. Since some genes can form conserved clusters~\cite{spring2015}, treating all gene adjacencies as fragile regions can lead to a huge bias in the estimation. 
Our estimator is based only on breakpoints (but not on conserved adjacencies) across two genomes, and so it avoids this issue. 

We further remark that our method is based on accurate estimation of the ratio $\frac{d}{b}$. At the same time, 
in comparative genomics studies the ratio $\frac{2d}{b}$ is known as the \emph{breakpoint reuse rate}~\cite{PevznerTeslerPNAS03,Attie2011}, 
representing the average number of breakpoint ``uses''  by rearrangements in the course of evolution between two genomes. 
Our parsimony phase condition can therefore be restated as the breakpoint reuse rate being below $1.58$. 
We argue however that the conventional definition of the breakpoint reuse rate does not accurately capture its meaning, 
in the same way how the rearrangement distance does not quite capture the meaning of evolutionary distance. 
This inspires a notion of the \emph{true breakpoint reuse rate} defined as $\frac{2k}{b}$ (i.e., with the true evolutionary instead of rearrangement distance in the numerator), 
which can be easily estimated with our method.

In further development of our method, we plan to address the even more accurate turnover fragile breakage model~\cite{Alekseyev10b}, where the set of fragile regions changes with time. We believe this model can better explain the cycle structure of the breakpoint graphs of real genomes (for a recent discussion of associated combinatorial aspects, see~\cite{AlexeevJCB2016}).

\bibliographystyle{bmc-mathphys} 
\bibliography{true_dist.bib}     

\end{document}